\documentclass[11pt,letterpaper]{article}

\usepackage{fullpage}
\usepackage{amsmath,amssymb,amsthm,dsfont}
\usepackage{tikz}
\usetikzlibrary{shapes.symbols,arrows}

\usetikzlibrary{shapes.symbols,arrows}
\usepackage{enumerate}

\usepackage{algorithm,algorithmic}

\newcommand{\In}{{\mathrm {In}}}

\newcommand{\cT}{{ (c,T)  }}

\newcommand{\oY}{\overline{Y}}

\newcommand{\IR}{\mathds{R}}
\newcommand{\IN}{\mathds{N}}

\renewcommand{\leq}{\leqslant}

\renewcommand{\geq}{\geqslant}

\newcommand{\ignore}[1]{}

\newtheorem{thm}{Theorem}
\newtheorem{prop}[thm]{Proposition}
\newtheorem{lem}[thm]{Lemma}
\newtheorem{cor}[thm]{Corollary}
\newtheorem{defi}[thm]{Definition}

\begin{document}

\title{Synchronization in Dynamic Networks}

 \author{%
 Bernadette Charron-Bost\textsuperscript{1} 
  \and Shlomo Moran\textsuperscript{2}}
 \date{\textsuperscript{1} \'Ecole polytechnique, 91128 Palaiseau, France\\
 \textsuperscript{2} Department of Computer Science, Technion, Haifa, Israel 32000\\~\\
 \today
  }




\maketitle

\begin{abstract}%
	In this article, we study algorithms for dynamic networks with asynchronous start, i.e.,  each node may start 
		running the algorithm in a different round. Inactive nodes transmit only heartbeats, which contain no information 
		but  can be detected by active nodes. 
	We make no assumption on the way the nodes are awakened, 
		except that for each node $u$ there is a time~$s_u$ in which it is awakened and starts to run the algorithm.
	The identities of the nodes are not mutually known, and the network size is unknown as well.

	We present synchronization algorithms, which guarantee that after a finite number of rounds, 
		all nodes hold the same round number, which is incremented by one each round thereafter. 
	We study the time complexity and message size required for synchronization, and specifically 
		for {\em simultaneous} synchronization, in which all nodes synchronize their round numbers 
		at exactly the same round. 

	We show that there is a strong relation between the complexity of simultaneous synchronization 
		and the connectivity of the dynamic graphs: With high connectivity which guarantees that messages 
		can be broadcasted in a constant number of rounds, simultaneous synchronization by all nodes can 
		be obtained by a deterministic algorithm  within a constant number of rounds,  and with messages 
		of constant size.
	With a weaker connectivity, which only guarantees that the broadcast time is proportional to the network size,
		our algorithms still achieve simultaneous synchronization, but within linear time and long messages.

	We also discuss how information on the network size and randomization may improve synchronization algorithms,  
		and show related impossibility results.
\end{abstract}

\ignore{macros

\EventEditors{John Q. Open and Joan R. Acces}
\EventNoEds{2}
\EventLongTitle{42nd Conference on Very Important Topics (CVIT 2016)}
\EventShortTitle{CVIT 2016}
\EventAcronym{CVIT}
\EventYear{2017}
\EventDate{December 24--27, 2016}
\EventLocation{Little Whinging, United Kingdom}
\EventLogo{}
\SeriesVolume{42}
\ArticleNo{23}
endignores macros}

\section{Introduction}\label{sec:intro}
We study distributed algorithms for {\em dynamic networks} over an arbitrary finite set of nodes~$V$
	that operate in synchronized rounds, communicate by broadcast messages, and 
	in which the inter-node connectivity may change each round of communication.
The node identities, and even the cardinality of the set $V$,  are not mutually known.

In previous works it was typically assumed that algorithms in such dynamic networks are started simultaneously by all nodes, and consequently that all nodes share the true round number (which is incremented by one each round). In this paper we relax this assumption, and consider a model in which round numbers are unknown to the nodes, and further that each node may start running the algorithm in a different round. 
This relaxation is natural in environments with no central control which monitors the nodes activities. 
To make our results more general, we do not make any assumption on the way a node may become active 
	and start running the algorithm, except that eventually all nodes are active. 
	
In this model, we study the basic question of  synchronizing the network, in the sense that 
	we wish to ensure that eventually all nodes share the same round number.
More specifically, we will focus on the following two levels of synchronization:
	(a)   implementing local round counters that are eventually all equal, and
	(b)  synchronizing the nodes themselves -- and not only their round counters -- i.e.,  detecting
	the synchronization of the local round counters	simultaneously.
	
Simultaneous synchronization  (b)  can be useful in various situations such as
	real-time processing (where processors have to carry out some external actions simultaneously),
	distributed initiation (to force nodes to begin some computation in unison), or   distributed termination
	(to guarantee that nodes complete their computation at the same round).
It actually coincides with the {\em Firing Squad problem}~\cite{BL87, CDDS85, CD91}:
	a node {\em fires} when it detects synchronization of the round counters. In the context of clock synchronization, our results  imply conditions under which a simultaneous {\em phase synchronization} can be achieved in dynamic networks, given that the local clocks of the nodes have the same {\em frequency}  (see eg \cite{SS07}).

We investigate these two levels of synchronization in the context of dynamic networks: the communication topology
	may continuously and unpredictably change from one round to the next.
In particular, we do not assume any stability of the links.
We examine various connectivity  properties that hold, not necessary round by round, but  globally over finite periods of consecutive rounds~\cite{KKK02}.

Our synchronisation algorithms demonstrate  a strong relation between the possibility and cost of synchronizing 
	a network, and the time required to broadcast a message from each node in the network: 
	perhaps a bit surprisingly, when broadcasts from all nodes are possible within 
	a constant number of rounds, a simultaneous synchronization can be achieved by a simple algorithm 
	within a constant number of rounds, using messages of constant size. 
When broadcast time is linear in the network size~$n$, we still achieve simultaneous synchronization in time 
	which is proportional to broadcast time (i.e., linear time), but with messages of size $\Omega(n\log(n))$. 
When broadcast from each node is possible but there is no bound on the number of required rounds, simultaneous 
	synchronization is not possible, but the simple synchronization of round counters is still achievable within finite time. 
	
We then study models in which some bound~$N$ on the network size is  known.
We present there few impossibility results and algorithms, including a randomized algorithm that assumes an oblivious adversary and performs simultaneous synchronization in linear time with high probability, but with messages which are considerably shorter than the ones used by our deterministic algorithm for the same task.

\noindent{\bf{Related work.}} Synchronization problems in distributed systems have been extensively studied,
	but most of works assumed a fixed topology~\cite{Tel00} or a complete graph and at most $f$ faulty nodes~\cite{Lyn96},
	i.e.,  a {\em fixed\/} core of at least $n-f$ nodes.
	
Our work is closely related to
 the article by Kuhn et al.~\cite{KLO10} on distributed computation in dynamic networks:
	some of our results are based on the  generalization to asynchronous starts of their approach for counting the size of the network.
In turn, our algorithms provide solutions to distributed computations with	asynchronous starts.
 In ~\cite{Wat14}   a different model of networks with asynchronous start is studied, in which
	  inactive nodes do not submit any signal, and hence, unlike in our model, their existence cannot be detected by active node - a property which is essential for our results. 

	\section{ The Model}\label{sec:model}

	We consider a networked system with a {\em fixed} set of $n$ nodes.
	Nodes have  unique identifiers and the set of the $n$ identifiers is denoted by $V$.
	The identities of the nodes are not mutually known, and the network size is unknown as well.

	Each node is initially {\em passive}: it is part of the network, but sends only heartbeats -- that we 
		call {\em null} messages -- and does not change its state. 
	Upon the receipt of a special signal, it  becomes {\em active}, sets up its local variables  (with its initial state), 
		and starts executing its code.

	Execution proceeds in {\em synchronized rounds\/}: in a round $t~(t= 1,2\dots$), each node, be it active or passive,
		attempts to send messages to all nodes,
		receives messages from some nodes, and finally goes to a new state and proceeds to round $t+1$. 
	The round number $t$ is used for a reference, but is unknown to the nodes.
	Synchronized rounds are communication closed in the sense that no node receives messages in round~$t$ that are sent 
		in a round different from~$t$.

	Communications that occur at round~$t$ are modeled by a directed graph~$G(t)=(V, E_t)$ 
		that may change from round to round in dynamic networks~\cite{CFQS11:TVG}.
	We assume a self-loop at each node in all the graphs $G(t)$ since any node can communicate with 
		itself instantaneously.

	In each execution, every node~$u$ is  assumed to receive a unique start signal  at the  {\em beginning} of some round~$s_u$.
	Each execution of the entire system is thus determined by the list~$\left (s_u \right )_{u \in V}$ of rounds 
		at which nodes become active, by the collection of  initial states, 
		and by  the sequence of directed graphs~$\left (G(t) \right )_{t\in\IN}$, that we call a {\em dynamic graph}.

	The way start signals are generated is left  arbitrary: they could be sent by an external oracle (environment), or they could
		be generated endogenously as in the case of diffusive computations initiated by a subset of nodes.
	Similarly, the sequence of directed graphs can be decided ahead of time or, endogenously as in 
		{\em influence systems}~\cite{Cha12}.


	 \subsection{Paths and broken paths in a round interval}

	 We now fix some notation and introduce some terminology that will be used throughout  this paper.
	 First, let us fix an execution of an algorithm with the list of rounds~$\left (s_u \right )_{u \in V}$ at which nodes become active
		and the dynamic communication graph~$\left (G(t) \right )_{t\in\IN}$.

	If  $x_u$ is a local variable of node~$u$, then  $x_u(t)$ denotes the value of $x_u$ 
		at the beginning of round~$t$. 
	Thus $x_u(t)$ is undefined for $t<s_u$.
	We let $G^*(t) = (V, E^*_t) $ denote the directed graph of edges that  transmit non-null messages at round~$t\/$: 
		$(u,v)\in E^*_t$  if and only if
		 it is an edge of $G(t)$ and $u$ is active at round~$t$. 
	We denote the sets of $u$'s incoming neighbors  (in-neighbors for short) in the directed graphs $G(t)$ and $G^*(t)$
		by $\In_u(t)$ and $\In_u^*(t)$, respectively.

	We recall that the {\em product} of two directed graphs $G=(V,E)$ and $H=(V,E')$, denoted $G \circ H$, is 
		the directed graph  with the set of nodes $V$ and with an edge $(u,v)$ 
		 if there exists~$w\in V$ such that $(u,w)\in E$ and $ (w,v) \in E'$.
	For $t' > t\geq 1$, we let $ G(t:t') = G(t) \circ G(t+1)\circ \dots \circ G(t')$, 
		and by convention, $G(t:t)=G(t) $. 
	Similarly, $ G^*(t:t') = G^*(t) \circ G^*(t+1)\circ \dots \circ G^*(t')$.

	Let $\In_u(t  :t' )$ and $\In_u^*(t  :t' )$  denote the sets of $u$'s in-neighbors in $G(t:t')$
		and in $G^*(t:t')$, respectively.
	A directed edge $(v,u)$ of $G(t:t')$ corresponds to a non-empty set  of {\em dynamic paths} of the form 
		$P= (v_0=v,v_1,\dots,v_m=u)$, where $m = t'-t+1$ and 
		$(v_k,v_{k+1})$ is an edge of $G (t+k)$   for each $k=0,\ldots,m-1$. 

	We will say that the dynamic path $P$ is {\em broken} if one of the edges of $P$ carries a null message,
		i.e., $t+k > s_{v_k}$ for some $k \in \{0, \ldots, m-1 \}$.
	For brevity, we will use the terminology of {\em $v {\sim } u $ path} and {\em $v {\sim } u $ broken path} in the {\em round interval 
		$ [t ,t'] $}. 

	\subsection{A hierarchy of synchronization problems}

	Let $A$ be an algorithm with an integer variable $r_u$ for each node $u$,
		which are aimed at simulating synchronous round counters.

	\begin{description}
	\item[Synchronization:] The algorithm $A$ {\em achieves synchronization} if in each execution
	            of $A$ with the start signals   $(s_u)_{u\in V}$,  from some round $t_{synch} \geq \max_{u \in V} (s_u) $ 
	            onward, the $r_u$ counters are incremented by 1 in every round and are all equal, i.e., for every $t \geq t_{synch}$, 
	       \begin{enumerate}
		\item $ r_u (t+1) = r_u(t) + 1 $ 
		\item $ \forall u, v\in V, \ r_u (t) = r_v(t) $.
		\end{enumerate}  
	\end{description}
	In the following each node $u$ is equipped with an additional boolean variable $synch_u$
		initialized to $false$ at round $s_u$.

	\begin{description}
	\item[Synchronization detection:] 
	The algorithm $A$ {\em achieves  synchronization detection} if it achieves synchronization,
		and in addition it guarantees that each node eventually detects that the network is synchronized, i.e., 	
		\begin{enumerate}\setcounter{enumi}{+2}
		\item $ \forall u \in V, \ synch_u (t) = true \Rightarrow t \geq t_{synch}  $ 
		\item $ \forall u \in V,  \ \exists t_u \in \IN , \  \forall t \geq t_u, \  synch_u (t) = true $.
		\end{enumerate}  

	\item[Simultaneous synchronization detection:] 
	The algorithm $A$ {\em achieves simultaneous synchronization detection} if all nodes detect synchronization
		simultaneously, i.e.,	
		\begin{enumerate}\setcounter{enumi}{+4}
		\item $ \forall u, v \in V,  \  \forall t \geq \max\{s_u,s_v\}, \ synch_u (t) = synch_v (t) $. 
		\end{enumerate} 
	\end{description}
	Note that the latter condition of simultaneity actually corresponds to the classical  {\em Firing Squad} problem~\cite{CDDS85, BL87, CD91} 
		(i.e., all nodes can fire when the $synch$ variables are set to $true$).

	\subsection{Completeness and connectivity of dynamic graphs }

	In this paper we consider  the following connectivity conditions in dynamic graphs.

	\begin{defi}
	Let $T$ be a positive integer.
	We say that the dynamic graph $\left( G(t) \right)_{t \geq 1}$ is  {\em $T$-complete } if 
		for every $t\geq 1$, the graph $ G(t: t+T-1) $ is  complete.
	\end{defi}
	Informally, $T$-completeness of a dynamic graph means that a message initiated by any node $u$ in any round $t$ can be broadcasted to all other nodes within $T$ rounds.

	We next define dynamic graphs which enable broadcasts in linear time. 
	For that, we  first introduce the concept of {\em in-connectivity} for directed graphs.
	\begin{defi}
	Let $G =(V,E) $ be a directed graph with at least two nodes and let $c <|V| $ be a positive integer.
	We say that $G$ is {\em  $c$ in-connected} if for any non-null subset $S$,
		the following holds:
		$$ | \Gamma_{\mathrm{in}} (S) \setminus S | \, \geq \min (c, |\overline{S} | )$$
		where $ \Gamma_{\mathrm{in}} (S) $ denotes the set of in-neighbors of $S$ in $G$, and $\overline{S}=V\setminus S$.
	\end{defi}

	Note that  $G$ is $|V|-1$ in-connected iff it is complete and that
		it is 1 in-connected iff it is strongly connected.

	One can define in an analogue way the $c$ out-connectivity of a directed graph. The following shows that these definitions are equivalent.
	\begin{prop}
	Let $G$ be a directed graph. 
	For each positive integer $c$  it holds that $G$ is $c$ in-connected iff it is $c$ out-connected	
	\end{prop}
	\begin{proof}
		We will prove that if $G$ is $c$ in-connected then it is also $c$ out-connected. The proof of the other direction is essentially identical.

		Assume for contradiction that $G$ is $c$ in-connected but not $c$ out-connected. Then there is a proper subset $S$ of $V$ s.t.   $|\Gamma_{\mathrm{out}} (S) \setminus S |=b<c$, and $|\overline{S}|>b$.

	Let $X=\Gamma_{out}(S)\setminus S$, 
	and let $R=\overline{S}\setminus X$. Since  $|X|=b$ and $|\overline{S}|>b$,  $R$ is not empty. By definition of $R$, there is no edge from a node in $S$ to a node in $R$, meaning that $\Gamma_{\mathrm{in}} (R) \setminus R \subseteq X$. Hence  $|\Gamma_{\mathrm{in}} (R) \setminus R|\leq b$. Since $\overline{R}=X\cup S$ contains  $b+|S|>b$ nodes, this contradicts the assumption that $G$ is $c$ in-connected. 
	\end{proof}

	\begin{defi}
	Let $c,T$ be two positive integers.
	We say that the dynamic graph $\left( G(t) \right)_{t \geq 1}$ is  {\em $\cT$ in-connected } if 
		for every $t\geq 0$, the graph $ G(t: t+T-1) $ is  $c$ in-connected.
	\end{defi}

	The  $\cT$ in-connectivity implies that a message initiated by any node $u$ in any round~$t$ can be broadcasted 
		to all other nodes within $\lceil\frac{T}{c}n\rceil$ rounds (for $c=T=1$, this is implied by a basic inequality 
	 	on the length of message chains -- e.g., Lemma 3.2 in~\cite{KLO10} --  and the generalization for arbitrary 
		$c$ and $T$ is straightforward).

	Finally, we present our weakest connectivity assumption which can be seen as $\infty$-completeness:
		a message initiated by a node $u$ in round~$t$ 
		can be broadcasted to all other nodes, but the time required for this broadcasting is unbounded.

	\begin{defi}
	A dynamic graph is said to be {\em eventually  strongly connected} if for every $t\geq 1$, 
		there exists $t'\geq t$ such that the graph $G(t:t')$ is  strongly connected.
	\end{defi}

	In the following we will present algorithms which achieve synchronization in the above models. 
	In the first two models, simultaneous synchronization detection is achieved within constant and linear broadcast time, respectively,
		but with substantially different messages sizes. 
	We will start with the last, weakest model.

	\section{Synchronization with Unbounded Broadcast Time}

	In this section, we show how the nodes in any  dynamic graph that is eventually strongly connected (and hence 
		guarantees broadcasting in finite but unbounded number of rounds)
		can eventually synchronize despite asynchronous starts.
	The synchronization algorithm (Algorithm~\ref{algo:ES}) is simple and does not use identifiers:
		 nodes may be assumed to be  anonymous and to  have computation and storage capabilities 
		 that do not grow with the network size~\cite{HOT11}. 

	First let us introduce one notation for the pseudo-codes of all our algorithms: we use $M_u^*$ to  denote the multiset 
		of non-null messages received by $u$ in the current round. 	 
	Thus $M_u^*$  at round~$t$ is the multiset of messages sent to $u$ by the nodes in $\In^*_u(t)$.
	If non-null messages are vectors of same size, then ${M_u^*}^{(i)} $ denotes the multiset of the $i$-th entries
		of the messages in $M_u^*$.

	\begin{algorithm}[h]
		\small
		\begin{algorithmic}[1]
			\REQUIRE{}
			   \STATE $r_u \in \IN $, initially $0$
			\ENSURE{}
			\STATE send $\langle r_u  \rangle $  to all processes and receive one message from each in-neighbor
			\IF{at least one received message is null}
			      \STATE  $r_u \gets 0$
			\ELSE
				\STATE $r_u \gets 1 + \min_{r\in M_u^*} (r) $      
			\ENDIF

		\end{algorithmic}
		\caption{Algorithm for synchronization  }
		\label{algo:ES}
	\end{algorithm}

	\begin{thm}\label{thm:synchesc}
	Algorithm~\ref{algo:ES}  achieves synchronization in any  dynamic graph that is eventually strongly connected.
	\end{thm}

	\begin{proof} 
	For any round $t\geq s_{\max} = \max_{u \in V} (s_u)$, we let 
		$$ W(t) = \big\{  u \in V \ : \   r_u(t) = \min_{v\in V}   \big( r_v(t)\big ) \big \} \enspace.$$
	Because of self-loops, we have $ W (t) \subseteq W (t+1) $.
	Moreover, if $W (t)$ has an outgoing edge in the directed graph $G(t)$, then $ W (t) \neq W (t+1)$.
	Eventual  connectivity of the dynamic graph ensures that from
		some round onward, we have  $W(t) = V$, which implies the theorem.
	\end{proof}

	We now state two useful lemmas about the way the local round counters $r_u$'s
	  evolve in dynamic graphs.

	\begin{lem}\label{lem:tls}
	Assume that $t<t'$ and $s_u\leq t'$.
	Then $r_u(t')$ is defined and:
		\begin{enumerate}
		\item \label{tls1} 
		If there exists a broken path ending at $u$ in the round interval $[t , t'-1]$, then $r_u(t') \leq   t' - t     - 1 $.
	 	\item 	\label{tls2}
		Otherwise, for every $v\in In_u(t:t'-1)$ it holds that $r_v(t)$ is defined  and    $r_u(t') \leq r_v(t) + t' -t$.
		\end{enumerate}
	\end{lem}

	\begin{proof} 
		\begin{enumerate}
		\item    
		Let  $P=(v_t=v,v_{t+1},\dots,v_{t'}=u)$ be the assumed broken path, and
		let $(v_{i-1}, v_{i})$,  be the last edge in $P$ which carries a null message ($ t+1 \leq i \leq t'$). 
		Then  $v_i$ is active at round$ i $, and by line 4 in Algorithm~\ref{algo:ES}, $r_{v_i}(i)=0$ .
		By easy induction, for $k =i+1, \dots ,t'$,  it holds that $r_{v_k}(k) \leq k -i\leq k -t-1$.
		Substituting $k=t'$, we obtain that $r_u(t') \leq t' - t -1$.

		\item   If there is no such broken path, then for each $v {\sim } u$ path $P$  as above,  no edge in $P$ 
		carries a null message, i.e.,  node~$v_k $ is active at round~$k$ for  each $ k=t,\ldots, t'-1 $.\\
		By line~6 in Algorithm~\ref{algo:ES} and a straightforward induction, for $k =t+1, \dots ,t'$,  
		it holds that $r_{v_k}(k) \leq  r_v (t)  +k -t$.
		Substituting $k=t'$, we obtain that $r_u(t') \leq  r_v (t)  + t' -t $.
		\end{enumerate}
	\end{proof}

	\begin{lem}\label{lem:smax}
	For every node~$u$ and at every round~$t \geq s_{\max} = \max_{v \in V} (s_v)$, we have $r_u(t) \geq t - s_{\max}$. 
	Moreover, if $t \geq s_{\max} +1$ and $\/\In_u(s_{\max} :t -1)$ contains a node $v$ such that $s_v = s_{\max}$, 
		then $r_u(t) = t - s_{\max}$. 
	\end{lem}

	\begin{proof} 
	By the definition of $s_{max}$, for each node $u$ we have $r_u(s_{max})\geq 0$, and an easy induction 
		on  $t\geq s_{\max}$ shows that $r_u(t) \geq t - s_{\max}$.

	If $\/\In_u(s_{\max} :t -1)$ contains a node $v$ such that $s_v = s_{\max}$, then  there is
		a $v {\sim } u $ path  in the round interval $ [s_{max} , t -1] $.
	Since all nodes are active on round $s_{\max}$ onward, only non-null messages are sent and all the  $v{\sim }u$ paths in the round interval $[s_{max},t -1]$
		are non broken.  
	The opposite inequality $r_u(t) \leq t - s_{\max}$ now follows from Lemma~\ref{lem:tls}.\ref{tls1} and $r_v \big( s_{\max} \big) =0$.
	\end{proof}

	\section{Simultaneous Synchronization Detection with Constant Time Broadcasting}

	We now show that the synchronization of the round counters can be detected
		 in any $T$-complete dynamic graph.
	Synchronization detection can be achieved simultaneously by all the nodes in $O(T)$ time 
		using only  $O(\log(T))$ bits per message.

	\begin{algorithm}[h]
		\small
		\begin{algorithmic}[1]
			\REQUIRE{}
			\STATE $r_u \in \IN $, initially $0$ 
			\STATE $synch_u \in \{true, false\}$, initially $false$ 
			\ENSURE{}
			\STATE send $\langle r_u  \rangle $  to all processes and receive one message from each in-neighbor
			\IF{at least one received message is null}
				\STATE  $r_u \gets 0$
			\ELSE
				\STATE $r_u \gets 1 + \min_{r\in M_u^*} (r) $
			\ENDIF
			\IF{  $ r_u \geq T$}
				\STATE $synch_u \gets true $
				\ENDIF
		\end{algorithmic}
		\caption{ Simultaneous synchronization detection with $T$-completeness}
		\label{algo:complete}
	\end{algorithm}

	\begin{thm}\label{thm:completeness}
	Algorithm~\ref{algo:complete} achieves Simultaneous Synchronization  in any execution  on a $T$-complete dynamic graph.
	Specifically, all nodes   detect the synchronization 
	 	of the $r_u$ counters   exactly $T$ rounds after all nodes have become active. 
	\end{thm}

	\begin{proof}
	First, observe that by the first claim in Lemma~\ref{lem:smax}, the condition in line~9, namely 
		$ r_u \geq T  $, eventually holds at each node~$u$.
	Moreover, $r_u$ may increase by at most 1 in every round, and thus
		hence there exists at least one round at which  $r_u$ is equal to $T$.

	Let $t_0$ be the first round at which some node~$u$ sets its variable~$r_u$  to $T$. 
	Then $r_u(t_0 +1)=T  $ and 
		\begin{equation}\label{eq:leq}
		 	\forall   v\in V ,  \ \ r_v(t_0 +1)\leq T   \enspace.
		\end{equation}
	Note also that  $t_0 \geq  r_u(t_0)$, and hence   $t_0 -T +1 \geq 1$.

	By $T$-completeness,  $ \In_u(t_0-T +1 : t_0 )=V$.
	Since $r_u(t_0+1)=T$, Lemma \ref{lem:tls}.\ref{tls1} implies that there is no broken path ending at node~$u$ 
		in the interval $[ t_0-T+1 , t_0 ]$.
	Hence,  $s_{max}=\max_{v \in V} (s_v)\leq t_0-T+1$ and $ \In^*_u(t_0 - T+1: t_0 ) = \In_u(t_0 - T+1: t_0 ) = V$.

	In particular, $\In^*_u(t_0 - T+1: t_0 ) $ contains the latest woken-up nodes. Hence, by Lemma~\ref{lem:smax}, 
		$T=r_u(t_0 +1) = t_0 - s_{\max} + 1 $
		and for every node $v\in V$, $ r_v(t_0 +1) \geq r_u(t_0 +1) = T $.
		Since by (\ref{eq:leq}) above $ r_v(t_0 +1)\leq T$, we get that $r_v(t_0+1)=T$ for all $v\in V$. By the definition of $t_0$, we also have that  $r_v(t')<T$ for $t'\leq t_0$ and for all $v\in V$. Thus all nodes set their local counters to $T$ at round~$t_0$ for the first time.
	 \end{proof}

	Let the {\em synchronization time} be the number of rounds from the time the last node is waked up till (simultaneous) 
		synchronization is achieved. 
	Then the synchronization time of any execution of Algorithm~\ref{algo:complete} is in $O(T)$, and it uses messages 
		of size $O(\log T)$.
		
		\section{Simultaneous Synchronization Detection with Linear Time Broadcasting}

		We now present Algorithm~\ref{algo:connectivity} for the simultaneous detection of the counters~$r_u\,$
			in dynamic graphs that are not $T$-complete, but still enjoy good connectivity, namely $\cT$ in-connectivity, which as mentioned earlier enables broadcasting in time which is linear in the number of nodes.
		As opposed to the previous algorithm, Algorithm~\ref{algo:connectivity} requires unique node identifiers and long messages.
		Indeed, each node $u$ maintains a variable $HO_u$ that contains the identifiers of all the active nodes 
			of which $u$ has heard of since it became active, and broadcasts $HO_u$ in each round.

		\begin{algorithm}[h]
			\small
			\begin{algorithmic}[1]
				\REQUIRE{}
				\STATE $r_u \in \IN$, initially $0$ 
				\STATE $synch_u \in \{true, false\}$, initially $false$ 
				\STATE $HO_u \subseteq V $, initially $\{ u \}$ 
				\ENSURE{}
				\STATE send $\langle r_u,HO_u  \rangle $  to all processes and receive one message from each in-neighbor
				\IF{at least one received message is null}
					\STATE  $r_u \gets 0$
				\ELSE
					\STATE $r_u \gets 1 + \min_{r\in {M_u^*}^{(1)}} (r) $
				\ENDIF
				\STATE $HO_u \gets \cup_{HO \in {M_u^*}^{(2)}} HO  $ 
				\IF{  $ | HO_u | \,  \leq    \left\lceil    \frac{c }{ T}  \left( r_u +1 \right  )\right\rceil  -c \, $}
					\STATE $synch_u \gets true $
					\ENDIF
			\end{algorithmic}
			\caption{ Simultaneous synchronization detection  with in-connectivity}
			\label{algo:connectivity}
		\end{algorithm}

		The correctness proof of the algorithm relies on the following inequality, which can be seen as a generalization 
		   of a basic inequality established 
		 	for {\em global round numbers}~$t$   (e.g., Lemma 3.2 of~\cite{KLO10}) 	to {\em local round counters}~$r_u$. 

		\begin{lem}\label{lem:cT}
		In each execution of Algorithm \ref{algo:connectivity} in a $\cT$ in-connected dynamic graph,  
			for each node $u$ and each round $t \geq s_u$, it holds that 
			\begin{equation}  | HO_u(t) | \, \geq  \min \left(  (1-c) +     \frac{c }{ T}   \left( r_u(t) +1 \right )  , n  \right ) \enspace . \label{eqn:l8}
			\end{equation}
		\end{lem}

		\begin{proof}

		For  $t = s_u $, we have $r_u(t)=0$ and  the proof is immediate.

		Suppose now that $t\geq s_u +1$, and let $q$ and $r$ two nonnegative integers such that 
			$ t = qT + r $  with $1 \leq r  \leq T$.
		By induction, we construct a sequence of $q+1$ sets of nodes  $ S_0,S_1, \dots, $  as follows:
			\begin{enumerate}
			\item $S_0 =  \{ u \} $.
			\item Suppose that $S_0,  \dots, S_i$ are defined and  $i < q $.
			We let $H_{i+1}=G \big( t-(i+1)T   : t-iT-1 \big)$. 
			We distinguish three cases.
				\begin{enumerate}
				\item	$H_{i+1}$ contains  no edge $(w,v)$ such that  $w\not\in S_i$ and $v\in S_i$. 
				Then the construction stops.
				\item
				$H_{i+1}$ contains an edge $(w,v)$ such that  $w\not\in S_i$ and $v\in S_i$, and that corresponds to a $w {\sim } v$ broken path
				in the round interval $[ t-(i+1)T  , t-iT-1 ]$.   
				Then the construction stops.
				\item Otherwise, we let   $S_{i+1}=  \In_u ( t-(i+1)T   : t )$.
				\end{enumerate}
			\end{enumerate}
		Let $S_k$ denote the last set in the sequence. \\
		A straightforward induction shows that for each node $v\in S_i$ there is a $ v {\sim } u $ path  in the 
			round interval $ [ t- iT ,  t ]$  which is not broken.\\

		Using the $(c,T)$ connectivity of G, an easy induction shows that the cardinality of each set $S_i$ is at least $c\,i+1$,
			except for the case when the construction is terminated by (a) above, 
			in which case $|S_k|$ may be smaller than $ck+1$.

		Similarly, by lines (3) and (10) of the algorithm and by the conditions (b) and (c) above, an easy induction shows that 
			$HO_u(t) $ contains every set $S_i$ for $0 \leq i \leq k$, and in particular $S_k\subseteq HO_u(t)$. 
		Hence if the construction is terminated by (b), then the cardinality of $HO_u(t)$ is at least $ c k  +1$.

		We now distinguish the following three cases:
			\begin{description}
			\item [Construction terminated by (a):] 
			By the $\cT$ in-connectivity of the dynamic graph,  this implies that $S_k = V$.
			It follows that $HO_u(t) =V$ and Lemma~\ref{lem:cT} trivially follows.
			\item [Construction terminated by (b):]
			We first observe that the assumed $w {\sim } v$ broken path in the round interval  $\big [ t-(k+1)T  , t- kT-1 \big ]$ can be extended 
			to a $w {\sim } u $ broken path in the interval $\big [ t-(k+1)T  , t \big ]$. 
			This implies by Lemma~\ref{lem:tls}.\ref{tls1} that
			$ r_u(t)\leq (k+1)T-1$ or equivalently that   $    k\geq \frac{r_u(t)+1}{T}-1$.  
		 Thus we get
		 	$$  |HO_u(t)|\geq ck+1\geq c\left (   \frac{r_u(t)+1}{T}-1\right )+1
				  	 =   (1-c) +     \frac{c }{ T}   \left( r_u(t) +1 \right ) \enspace .$$ 
			 \item [Construction terminated by $k=q\,$:] 
			 Thus we get   $  |HO_u(t)|\geq cq+1 $.
			 As an immediate consequence of Lemma~\ref{lem:tls}.\ref{tls1}  and~\ref{lem:tls}.\ref{tls2},
			 we have $r_u(t) \leq t-1 $.
			 With $r \leq T$, it follows that $r_u(t) \leq T(q+1) -1 $, and
			 $$  |HO_u(t)|\geq  (1-c) +     \frac{c }{ T}   \left( r_u(t) +1 \right ) $$ 
			 which also gives Lemma~\ref{lem:cT}   in this last case.
			\end{description}
		\end{proof}

		\begin{thm}\label{thm:cT}
		In any  $\cT$ in-connected dynamic graph, all nodes in Algorithm~\ref{algo:connectivity}  achieve simultaneous  
			synchronization detection.
		Synchronization is detected   in less than $ \left\lceil \frac{T}{ c}   \left( n-1 \right ) \right\rceil +T $ rounds  
			after all nodes have become active. 
		\end{thm}

		\begin{proof} 


		First, observe that by the first claim in Lemma~\ref{lem:smax} and the fact that the cardinality of each set~$HO_u$
			is at most  $n$,  the condition in line~11 eventually holds at each node~$u$. 
		Moreover, we easily check that at each round $t\geq s_u$, it holds that
			\begin{equation}\label{eq:HO}
			   HO_u  (t +1)   \subseteq \bigcup_{s \geq s_u} \In_u^* (s : t)    \enspace.
			\end{equation}

		Let $t_0$ be the first round at which the condition in line~11 holds at some node,
			and let~$u$ denote one node such that	
			\begin{equation}\label{eq:11}
			  | HO_u  (t_0 +1) | \,  \leq    \left\lceil    \frac{c }{ T}  \left( r_u (t_0 +1)  +1 \right  )\right\rceil  -c \,    \enspace.
			\end{equation}
		From Lemma~\ref{lem:cT}, we deduce that  $HO_u(t_0 +1) = V $.
		In particular, $HO_u(t_0 +1 ) $ contains the latest woken-up nodes. Let~$v$ denote one such node, ie $s_v= s_{\max}$.
		By (\ref{eq:HO}), there is a $v {\sim } u$ non broken path  in some round interval $[ s, t_0]$  with $s\geq s_u$.
		It follows that $s\geq s_v$.
		Thereby $t_0 \geq s_{\max} $ and $v \in \In_u^* (s_{\max} : t_0) $.
		This implies, by  Lemma~\ref{lem:smax}, that 
			$$r_u(t_0+1) = r_v(t_0+1)=t_0+1-s_{max}=\min_{w\in V}r_w(t_0+1) \enspace .$$
			 Using Lemma~\ref{lem:cT}=10 again we get that for all $w\in V$, $H_w(t_0+1)=V$. 
		Therefore the inequality (\ref{eq:11}) holds for all nodes in round $t_0+1$, and by the definition of $t_0$ this is the first round in which  this inequality holds for all nodes.

		\end{proof}

		Considering that $c$ and $T$ are constants,  the synchronization time of Algorithm~\ref{algo:connectivity}  
			is $O (n)$, and it uses messages of size $O ( n\log n )$.

		A close examination of the proof of Theorem~\ref{thm:cT}, shows that 
			each node actually computes the set $V$, and so its cardinality.
			In other words, Algorithm~\ref{algo:connectivity} solves the problem of {\em  counting} the network size
			despite asynchronous starts in any $\cT$ in-connected dynamic graph, and in particular in any continuously 
			strongly connected dynamic graph.

		This should be compared  with the impossibility result established by
			Wattenhofer~\cite{Wat14}, which shows that if inactive nodes do not transmit any signal, 
			then counting is impossible with asynchronous start. In other words,
			in the latter network model  passive nodes are not  considered as 
			part of the network:
			the set of nodes in the network is thus time-varying  while we assume a fixed set $V$ of nodes.

			\section{Synchronization with Bounds on Network Size}\label{sec:size1}

			In this section, we show that knowledge of bounds on the network size can sometimes improve our synchronization algorithms.

			\subsection{Simultaneous Synchronization Detection with linear broadcast time and short messages}\label{sec:size}

			We now show that knowledge on network size can be used for reducing the message size 
			 	required for linear-time simultaneous synchronization in $\cT$ in-connected dynamic networks.
			For the sake of simplicity, we assume $c=T=1$ , i.e., dynamic graphs are continuously strongly connected,
				but all the results can be obviously extended to the general case of in-connectivity.

			First, observe that any  connected dynamic graph with $n$ nodes is $(n-1)$-complete.
			Thus one immediate spinoff of Theorem~\ref{thm:completeness} is the following corollary, 
			which provides a solution to the simultaneous synchronization detection problem in
			strongly connected dynamic graphs when an upper bound $N$ on the network size
			is known by all nodes.

			\begin{cor}\label{cor:N}
				If nodes have an upper bound $N$ of the network size, simultaneous synchronization detection can be 
				achieved in any  continuously strongly connected dynamic graph 
				in $N$ rounds after all nodes have become active using only $O(log(N))$ bits per message.
			\end{cor}

			When $N$ is significantly larger than the network size $n$,  
			this solution to simultaneous synchronization detection may use much more than $n$ rounds.
			However,   Lemma~\ref{lem:cT} enables us to extend the randomized algorithm for approximate
			counting presented in~\cite{KLO10,Osh12} to the case of asynchronous 
			starts, by substituting the local round counters for the round numbers.
			As we shall show below, this yields a randomized algorithm which has only a loose bound $N$ on 
				the network size $n$, 
				sends  short messages, 
				and  with high probability enables all nodes to achieve simultaneous synchronization detection within $O(n)$ rounds.
			For this algorithm, it is assumed that the dynamic graph, and the wakeup times $s_v$, are managed by an {\em oblivious} adversary, which has no access to the outcomes 
				of the coin tosses made by the algorithm.

			The algorithm, denoted ${\cal A}_{N,\eta}$, depends on the two parameters $N$ and $\eta$ 
				where $N$ is an upper bound on the network size and $\eta$ is any real number in $[\, 0,1/2 \, [$.
			It works as follows: upon becoming active, each node $u$ generates $\ell $
				independent random numbers $Y_u^{(1)}, \dots,Y^{(\ell)}_{u}$, where the distribution of each $Y_u^{(i)}$
				is exponential with rate~1.
			At each round, any active node~$u$ first broadcasts the smallest value of the $Y_v^{(i)}$'s
				it has heard of for each index~$i$, and then computes from the  minimum values it received so far an estimation  
				$n_u$ of the cardinality of the set  of nodes it heard of.  
			It detects the synchronization of the round counters when  its round counter $r_u$ is sufficiently larger than $n_u$.

			The analysis of the algorithm ${\cal A}_{N,\eta}$ relies on the following lemma in~\cite{MAS06}
				which is  is an application of the Cram\'er-Chernoff method (see for instance
				\cite{BoucheronLugosiMassart2013}, sections 2.2 and 2.4). 

			\begin{lem}\label{lem:cramer}
			Let $I$ be a finite set of $\ell$-tuples of independent exponential variables all with rate~1:
				$I = \left \{     \left ( Y_1^{(1)} , \dots , Y_1^{(\ell)}   \right) ,   \dots   ,  \left (   Y_{| I |} ^{(1)} , \dots , Y_{| I |} ^{(\ell)}   \right)                      \right \} $,
				and let 
				$W := \frac{\ell}{\sum_{i=1}^{\ell} \min_{1 \leq j \leq | I | } Y_j^{(i)}} $.

			For any $\varepsilon \in [ 0, 1/2] $, we have 
				\begin{equation}
				  {\rm{Pr}} \left ( \left | \, W -  | I  | \, \right |  > 2 \, \varepsilon  |  I  | /3  \right )   \leq 2 \exp \left ( - \varepsilon^2 \ell  /27 \right )  \enspace.\label{equ:W}
				  \end{equation}
			\end{lem}

			It follows that  for a sufficiently large value of $\ell$, the value of~$n_u$ at the end of round~$t$ provides with high probability a good 
				approximation of  the number of active nodes that~$u$ has heard of so far.
			Then by Lemma~\ref{lem:cT}, it follows that with high probability, if $ n_u <  (1 - \varepsilon) r_u$	
				then all nodes are active and node~$u$ has heard of all.
			As for Algorithm~\ref{algo:connectivity}, we may conclude that with high probability,
				all the nodes  detect synchronization at the same round $t_0$ and make no false detection.

			We fix  the precision parameter of the randomized approximate counting algorithm  in~\cite{KLO10}  to $\varepsilon = 1/3$,
				and  choose $\ell =   \left\lceil   243 . (\ln 4N^2 - \ln \eta) \right\rceil  $ to guarantee a final probability of at least
				$1-\eta$ for successful executions (cf. below).
			Node~$u$ detects synchronization when $n_u < 2 r_u/3$ (see Algorithm~\ref{algo:random}). 

			One key point  of the randomized approximate counting algorithm in~\cite{KLO10} 
				lies in the fact that the algorithm still works when the random variables $Y_u^{(i)}$ are initialized with rounded 
				and range-restricted approximations
				of the initial random numbers of the above scheme.
			More precisely, if we round down each $Y_u^{(i)}$ to the next smaller integer power of $13/12 $, 
				then the probability that $n_v$ is a good approximation of the number of nodes $v$ have heard of (inequality~(\ref{eq:random1})
				below) still holds.
			By the definition of the exponential distribution, it is not hard to see that the  random variables $Y_u^{(i)}$ 
				are all within the range
			$ \left[ \eta/(4\ell N), \ln (4\ell N/\eta) \right]$ with high probability, namely 
				\begin{equation}\label{eq:rangepr}
				{\rm{Pr}} \left[ \forall u\in V, \ \forall i,  \  Y_u^{(i)} \in \left[ \eta/(4\ell N), \ln (4\ell N/\eta) \right] \right] > 1- \eta/2 \enspace, 
				\end{equation}
			 The two above transformations of rounding and range-restricting can then be combined and yield a collection of random variables
				denoted $\oY_u^{(i)}$.
			The correctness proof of the original algorithm with the exponential random variables  $Y_u^{(i)}$ is still valid
				when substituting $\oY_u^{(i)}$ for $Y_u^{(i)}$ except an additional probability of at most $\eta/2$ of
				unsuccessful executions in which the range~(\ref{eq:rangepr}) is violated.
			In addition, the number of distinct variables $\oY_u^{(i)}$ in that range   is $O(\log(N\eta^{-1})$, 
				hence each such variable can be represented using $O \left( \log \log (N / \eta)   \right)$ bits.	

			The reader is referred to~\cite{Osh12} for more details on the above scheme developed for the model 
				in which all nodes start at round 1 and hold the true round numbers. 
			Below, we  only present the  points in the correctness proof of  ${\cal A}_{N,\eta}$ that are specific
				to  the round counters $r_u$.

			\begin{algorithm}[h]
				\small
				\begin{algorithmic}[1]
					\REQUIRE{}
					\STATE $r_u \in \IN $, initially $0$ 
					\STATE $synch_u \in \{true, false\}$, initially $false$ 
					\STATE $\oY_u = \left(\oY_u^{(1)}, \dots,\oY^{(\ell)}_{u} \right) \in \IR^{\ell}$, initially
					$\ell =   \left\lceil   243 . (\ln 4N^2 - \ln \eta) \right\rceil $ rounded and range-restricted approximations of independent random numbers with 
					exponential distribution of rate 1.
					\STATE $n_u \in \IN $, initially $0$ 
					\ENSURE{}
					\STATE send $\langle r_u,\oY_u \rangle $  to all processes and receive one message from each in-neighbor
					\IF{at least one received message is null}
						\STATE  $r_u \gets 0$
					\ELSE
						\STATE $r_u \gets 1 + \min_{r \in {M_u^*}^{(1)} } (r) $
					\ENDIF
					\FOR{$i=1,\dots,\ell$}
					\STATE $\oY_u^{(i)} \gets  \min_{\oY^{(i)} \in {M^*_u}^{(i +1)}} \left( \oY^{(i)} \right) $
					\ENDFOR
					\STATE $n_u \gets \ell / \sum_{i=1}^{\ell} \oY_u^{(i)} $
					\IF{$ n_u  <  2\, r_u / 3$}
						\STATE $synch_u \gets true $
						\ENDIF
				\end{algorithmic}
				\caption{ The randomized algorithm ${\cal A}_{N,\eta}$  }
				\label{algo:random}
			\end{algorithm}

			\begin{thm}\label{thm:random}
			In any  execution of  Algorithm~${\cal A}_{N,\eta}$ on a dynamic graph that is continuously strongly connected with at most $N$ nodes,  with probability at least $1-\eta$
				 all nodes  
				detect the synchronization 
				of the $r_u$ counters simultaneously  in less than $2n$ rounds after 
				all nodes have become active. The algorithm uses
			 messages of size 
				$O \left ( \log (N /\eta). \log\log (N /\eta)  \right)$.
			\end{thm}

			\begin{proof}
			We fix any real number $\eta \in ] \, 0, 1/2 \, ]$.
			For any node~$v$ and at each round $t\geq s_v$, let $ \nu_v(t) $ denote the number of nodes that
				$v$ has heard of at the end of round~$t$.\footnote{Formally, 	$v$ has heard of $w$ at the end of round~$t$ if for some $s\leq t$,    there is a  path $(w=v_s,\ldots,v_{t}=v)$ where $(v_i,v_{i+1})\in G(i)$ and   both $v_i$ and $v_{i+1}$ are active in round $i,i=s,\ldots, t-1$.}\\ 
			For a continuously strongly connected graph, Lemma~\ref{lem:cT} then reads:
				\begin{equation} 
				 | \nu_v(t) | \, \geq  \min \left(  r_v(t +1) +1  , n  \right ) \enspace . \label{eq:nu}
				\end{equation}

			\noindent  Let $t_0$ be the  first round at which some round counter in the network is set to $n$ (the true network size).
			Using~(\ref{eq:nu}), we deduce that  all nodes are active at round $t_0$ and that 
				every node has heard of all at the end of round $t_0$,
				i.e., $t_0 \geq s_{\max}$ and $\nu_v(t_0) =n$.	
			Then the same argument as in the proof of Theorem~\ref{thm:cT} shows that all nodes have synchronized at the end of round $t_0$,
				namely
				\begin{equation}\label{eq:t0+1}
				 \forall v\in V   , \     \forall t\geq t_0+1, \   \ r_v(t ) = t - s_{\max}   \enspace.
				\end{equation}
			In particular, we have $t_0 = n -1 + s_{\max} $.	

			\noindent Since $\nu_v(t_0) =n$,  every node $v$ computes the {\em same} estimate $\tilde{n}$ of $n$
				at the end of round~$t_0$  in its variable $n_v$ and keeps this value for $n_v$ at all later rounds:
				\begin{equation}\label{eq:nuatt0+1}
				 \forall v\in V   , \    \forall t\geq t_0+1, \  \   \nu_v(t ) =    \tilde{n}  \enspace.
				\end{equation}
			It follows that  if the condition at line~15 holds at some node in round~$t \geq t_0$, 
				then  it holds at  all nodes in every round $t' \geq t$.

			\noindent By Lemma~\ref{lem:cramer}, for every node~$v\in V $ and every round $t$, $t_0 \geq t \geq s_v$, we have 
				\begin{equation}\label {eq:random1}
				{\rm{Pr}} \big[ \, \left | \, n_v(t +1)  -  \nu_v(t) \, \right |  > 2 \,  \nu_v(t)  /9  \,   \big] \leq 2 \exp \left ( - \ell  /243 \right ) \enspace .
				\end{equation}
			\noindent Using  (\ref{eq:nu}) and the inequality $r_v(t_0) \leq n-1$, we get that at each round $t$, $t_u \geq t \geq s_v$,
				the following implication holds:
				\begin{equation*}
					 \left | \, n_v(t+1)  - \nu_v(t) \, \right |   \leq  2 \, \nu_v(t)  /9  \, \Rightarrow \,  n_v(t +1) \geq 2 \, r_v(t+1)/ 3   \enspace .
				\end{equation*}

			\noindent Now observe that each node $v$ makes a true update to $n_v$ at line 14 of the algorithm at most $n-1$ times 
				(when the set of nodes it heard of strictly increases). 
			This implies, by the union bound, that the probability that node~$v$ does not detect synchronization  by round $t_0$ is at least
				$ 1 - 2( n -1) \exp \left ( - \ell  /243 \right ) $.
			Using  the union bound again and the upper bound  $N \geq n$, we obtain that the probability that no node detects synchronization  
				by round $t_0 $  when using the random variables $Y_u^{(i)}$ in the algorithm
				is thus at least $ 1 - 2 N (N-1) \exp \left ( - \ell  /243 \right )$.\\	

			\noindent Let $\theta$ denote the first round at which a node detects synchronization, i.e., the condition at line~15 holds 
				for the first time at round~$\theta $.	
			The above argument shows that the probability that $\theta \geq t_0$
				is at least $ 1 - 2 N (N-1) \exp \left ( - \ell  /243 \right ) $.\\	

			\noindent  Inequality (\ref{eq:random1}) at round~$t=t_0$ shows that with probability at least $1- 2 \exp \left ( - \ell  /243 \right ) $,
				it holds that 
				\begin{equation}\label {eq:random2}
				 \left | \, \tilde{n}  -  n  \, \right |   \leq  2 \,  n /9 
				  \end{equation}
				which implies that $ \tilde{n} \leq  11n/9$.
			Moreover equations~(\ref{eq:t0+1}) and (\ref{eq:nuatt0+1}) for $t = 2n + s_{\max} \geq t_0 +1 $ write:
				$$ \forall v\in V, \ \ r_v(2n + s_{\max} ) = 2n \ \mbox{ and } \  \nu_v(2n + s_{\max} ) =  \tilde{n} \enspace . $$
			It follows that the condition  at line~15 holds at round $ 2n + s_{\max} $, or equivalently $\theta \leq 2n + s_{\max} $, 
				with probability at least $1- 2 \exp \left ( - \ell  /243 \right ) $.\\

			\noindent  In conclusion,  with probability at least $1- 2 N^2 \exp \left ( - \ell  /243 \right )  \geq 1 - \eta/2 $,   it holds that 
				$$  n + s_{\max}   \leq  \theta \leq 2n + s_{\max}  \enspace . $$
			As explained above,  using  the approximated variables $\oY_u^{(i)}$ instead of $Y_u^{(i)}$ 
				results in an additional probability of at most $\eta/2$ of unsuccessful executions.
			This shows that  with probability at least $1- \eta$, all the nodes correctly and simultaneously detect synchronization 
				of their round counters by round $2n + s_{\max}$.

			\end{proof}

			The terminating variant of ${\cal A}_{N,\eta}$ 
				in which nodes stop executing their code after they have detected 
				synchronization (line~16) thus achieves simultaneous 
				synchronization detection with high probability:
				 running time 
				is in  $O(n)$ and messages are of size $O \left ( \log (N /\eta). \log\log (N /\eta)  \right)$.

			\subsection{Synchronization detection with unbounded broadcast time}\label{sec:esc}

			Theorem~\ref{thm:synchesc} demonstrates that nodes can eventually synchronize in any dynamic graph 
				that is eventually strongly connected. 
			In this section we show that this synchronization can be detected only if nodes know the exact network size $n$.
				%

			We first show that synchronization cannot be detected in such dynamic graphs
				even if it is given that the network size is either $n$ or $n+1$, for some fixed integer $n$ which is known to the algorithm. 

			\begin{thm}\label{thm:impossibility1}
			Synchronization of the $r_u$ counters cannot be detected  under the sole assumption of  eventual  connectivity.  
			This result holds even if it is given that  the network size is either $n$ or $n+1$, for some (fixed) $n$.
			\end{thm}

			\begin{proof}[Sketch of proof]
			For the sake of contradiction, suppose that there is an algorithm $A$ which detects synchronization
				in every execution on an eventually connected dynamic graph with   $n$ or $n+1$ nodes.

			Let $V$ be a set of cardinality $n+1$, let $u$ be a node in $V$, and let $W=V\setminus\{u\}$. 
			Then by our assumption, $A$ achieves synchronization detection in the execution $\cal{E}$ 
				over the dynamic graph $G$ in which each $G(t)$ is $K_W$, the complete directed graph over $W$, 
				and all nodes in $W$  are active from the first round.	
			Hence, there is some $t_0$ such that every node in $W$ has detected synchronization by round $t_0$.

			Now consider an execution $\cal{E}'$ over a dynamic graph $H$ in which 
				$H(t)=K^{u}_W$ for $1\leq t\leq t_0$ where $K^{u}_W$ denotes the directed graph over $V$ with the same edges as $K_W$, 
				and $H(t)=K_V$, the complete directed graph over $V$, for $t>t_0$. 
			Assume further that in $\cal{E}'$, $s_v=1$ for each $v\in W$, and $s_u=t_0+1$. 
			Then since nodes in $W$ cannot distinguish between $\cal{E}$ and $\cal{E'}$ during the first $t_0$ rounds, 
				they all incorrectly detect synchronization by round $t_0$ in $\cal{E'}$, i.e., before $u$ became active. 
			The proof is completed by noting that $H$ is eventually strongly connected over $V$.
			\end{proof}

			Interestingly, the latter impossibility result does not hold anymore when the exact size of the network is known.
			Indeed, thanks to its knowledge of $n$, each node~$u$ can detect  that all nodes are active.
			By Lemma~\ref{lem:smax},  its round counter~$r_u$ is then minimal amongst all the local round counters, 
				in which case $u$ is considered as {\em ready to synchronize}.
			Then~$u$ can determine when nodes in the network are all ready, that is to say  all round counters are minimal, 
				and thus are equal.

			For that, each node~$u$ maintains two sets of node identifiers, namely~$HO_u$ which is the set of active nodes
				$u$ has heard of so far, and $OK_u$ which is the set of nodes  that  $u$ knows to be ready to synchronize.
			The corresponding pseudo-code is given in Algorithm~\ref{algo:detectesc}.

			\begin{algorithm}[h]
				\small
				\begin{algorithmic}[1]
					\REQUIRE{}
					\STATE $r_u \in \IN $, initially $0$ 
					\STATE $synch_u  \in \{true, false\}$, initially $false$ 
					\STATE $HO_u \subseteq V $, initially $\{ p \}$ 
					\STATE $OK_u \subseteq V $, initially $\emptyset$ 
					\ENSURE{}
					\STATE send $\langle r_u,HO_u ,OK_u \rangle $  to all nodes and receive one message from each in-neighbor
					\IF{at least one received message is null}
						\STATE  $r_u \gets 0$
					\ELSE
						\STATE $r_u \gets 1 + \min_{r\in {M_u^*}^{(1)}} (r) $
					\ENDIF
					\STATE $HO_u \gets \cup_{HO \in {M_u^*}^{(2)}} HO $
					\STATE $OK_u \gets \cup_{OK \in {M_u^*}^{(3)}} OK  $
					\IF{  $| HO_u | \,  = n $}
						\STATE $OK_u \gets OK_u  \cup \{ u \}$
						\ENDIF
					\IF{  $| OK_u | \,  = n $}
						\STATE $synch_u \gets true$
						\ENDIF
				\end{algorithmic}
				\caption{ \ Synchronization detection with  eventual  strong connectivity}
				\label{algo:detectesc}
			\end{algorithm}

			We now show that synchronization cannot be detected simultaneously by all nodes of
				an eventually strongly connected network, demonstrating that with respect to synchronization,
				simultaneity is harder than detection in this network model.

			\begin{thm}\label{thm:impossibility2}
			Simultaneous synchronization detection is impossible in  eventually strongly connected dynamic graphs, 
				even if all nodes know the size of the network.
			\end{thm}

			\begin{proof}[Sketch of proof]

			By contradiction, suppose that there is an algorithm ${\cal A}$  that achieves simultaneous synchronization
				detection in any  eventually strongly connected dynamic graph. 

			Let $S$ denote the star directed graph centered at $u$,  
				and let  $S^ T $ be its transpose.
			Let $I= (V, E_I)$ the directed graph with only a self-loop at each node, i.e., $E_I = \{ (v,v) : v\in V \}$.

			We consider the execution of $A$ with  start signals all received in the first round, and the alternating 
				sequence of directed graphs $G = S, S^T\! , S, S^T\!, \dots $. 
			The dynamic graph $G$ is eventually strongly connected, and thus all nodes detect synchronization
				 at the same round~$t_F$.

			Now assume that  $G(t_F) =  S^T$ (the case $G(t_F) =  S$ is similar).
			From the viewpoint of any node $v\neq u$, $G$ is indistinguishable up to round~$t_F$ from the dynamic graph $G^{1 }$
				that is similar to $G$ except at round $t_F$ where $G^{ 1 }(t_F) =  I $.
			Hence all  nodes other than $u$  also detect synchronization at round$t_F$ with the dynamic graph $G^{ 1 }$.
			The same holds  for node$u$ since  the dynamic graph~$G^1$ is eventually strongly connected.

			By repeating this argument $t_F$ times, we show that all  nodes detect synchronization at round~$t_F$ in the execution
				of ${\cal A}$ with  start signals all received in the first round, and the dynamic graph 
				$G^{t_F }= I, \dots, I, S,S^T,S, S^T, \dots $. 

			From the viewpoint of any node $v\neq u$, the latter execution is indistinguishable up to round~$t_F$ from the
				execution with the same dynamic graph $G^{ t_F }$ and start signals all received 
				in the first round except the one received by $u$ at some round $s_u >  t_{F}$.
			In this execution, synchronization is detected earlier than $u$'s start, a contradiction.

			\end{proof}

			\section{Conclusion}

			In this paper, we defined  a model of distributed algorithms in which nodes do not start the algorithm simultaneously.  
			We studied algorithms in this model which synchronize the round counters of the nodes in a dynamic network, 
				where the network topology may change  each round, there is no information on the network size, 
				and node identities are not mutually known.
			As opposed to many models of dynamic networks developed for counting or consensus, 
				links are not supposed to be bidirectional, and we assume no stability of the network in time.

			We presented several  algorithms whose messages size and  time complexity  highly depend 
				on the connectivity of the topology.

			We also showed that with only eventual connectivity assumptions, synchronization detection is impossible
				unless nodes know the exact size of the network.

			Possible extensions of this work involve variations of the model of computation.
			For instance, it is interesting to know in which other models of connectivity, synchronization can be detected. 
			 It is also of interest to determine whether simultaneous synchronization detection is possible in an anonymous
				dynamic network where nodes have limited  storage capabilities and communicate through finite bandwith channels 
				as in~\cite{HOT11}.
			Our adaptation of the randomized algorithm of \cite{KLO10} provides an efficient  Monte Carlo solution for this problem, in the case of continuously strongly 
				connected networks.

			This raises another question concerning the role of leaders in a dynamic network: does the existence of a
				leader in an anonymous network may help for synchronization detection?
			Combined with our strategy for synchronization detection, the Metropolis method (see~\cite{NOOT09})
				yields  a deterministic  algorithm that achieves simultaneous synchronization detection in $O(n^3)$ rounds and that
				works in any anonymous dynamic network  with a leader and  a bidirectional connected topology.
			Unfortunately this algorithm uses messages of infinite size (nodes send real numbers) and do not tolerate 
				rounding.
			The existence of a deterministic algorithm for synchronization detection in polynomial time, with anonymous nodes and 
				bounded bandwith capacity is still an open problem.

			Also of interest are dynamic networks which enable the solution of the consensus problem with asynchronous start. 
			This problem could shed light on the relation between consensus algorithms and kernel agreement algorithms of \cite{CBS09}.

			\bibliographystyle{plain}

\end{document}